\documentclass[11pt,letterpaper]{article}

\usepackage{amsmath}
\usepackage{amsthm}
\usepackage{amssymb}
\usepackage{color}
\usepackage{graphicx}
\usepackage[boxruled]{algorithm2e}
\usepackage{amsopn}
\usepackage{mathtools}
\usepackage[margin=1.0in]{geometry}
\usepackage{authblk}

\newtheorem{assumption}{Assumption}

\newtheorem{lemma}{Lemma}
\newtheorem{remark}{Remark}
\newtheorem{theorem}{Theorem}

\newtheorem{definition}{Definition}

\begin{document}
\title{An Improved Approximation Algorithm for the Hard Uniform Capacitated $k$-median Problem}
\author{Shanfei Li\\Delft Institute of Applied Mathematics, TU Delft, The Netherlands\\shanfei.li@tudelft.nl}
\date{}

\maketitle
\begin{abstract}
In the $k$-median problem, given a set of locations,
the goal is to select a subset of at most $k$ centers so as to minimize the total cost of connecting each location to its nearest center.
We study the uniform hard capacitated version of the $k$-median problem, in which each selected center can only serve a limited number of locations.

Inspired by the algorithm of Charikar, Guha, Tardos and Shmoys,
we give a $(6+10\alpha)$-approximation algorithm for this problem with increasing the capacities by a factor of $2+\frac{2}{\alpha}, \alpha\geq 4$,
which improves the previous best $(32 l^2+28 l+7)$-approximation algorithm proposed by Byrka, Fleszar, Rybicki and Spoerhase violating the capacities by factor $2+\frac{3}{l-1}, l\in \{2,3,4,\dots\}$.
\end{abstract}

\section{Introduction}
In the capacitated $k$-median problem (CKM), we are given a set $N$ of locations (where a center can potentially be opened).
Each location $j\in N$ has a capacity $M$ (uniform capacities), and a demand $d_j$ that must be served.
Assigning one unit of the demand of location $j$ to center $i\in N$ incurs service costs $c_{ij}$.
We assume the service costs are non-negative, identity of indiscernibles, symmetric and satisfy the triangle inequality.
That is, $c_{ij}\geq 0, \forall i,j\in N$; $c_{ij}= 0$, if $i=j$; $c_{ij}=c_{ji}, \forall i,j\in N$ and $c_{it}+c_{tj}\geq c_{ij}, \forall i,j,t\in N$.
The objective is to serve all the demands by opening at most $k$ centers and satisfying the capacity constraints such that the total cost is minimized.
In this paper, we consider the \emph{hard} capacities and \emph{splittable} demands, that is, we allow at most one center to be opened at any location and each location can be served from more than one open center.
(In contrast, the \emph{soft} capacities allows that multiple centers can be opened in a single location.
In the \emph{unsplittable} demands case each location must be served by exactly one open center.)

CKM can be formulated as the following mixed integer program (MIP),
where variable $x_{ij}$ indicates the fraction of the demand of location $j$ that is served by location $i$,
and $y_{i}$ indicates whether location $i$ is selected as a center.
\begin{alignat}{2}
\min\quad&\sum_{i,j\in N} d_j c_{ij} x_{ij}\nonumber\\
\text{subject to:}\quad &\sum_{i\in N}{x_{ij}} =  1,\quad \forall j\in N; \quad &\sum_{j\in N}{d_j x_{ij}} \leq  M y_i, \quad \forall i\in N; \nonumber\\
&\sum_{i\in N}{y_i} \leq  k; \quad & 0 \leq x_{ij} \leq y_i, \quad \forall i, j\in N; \nonumber\\
&y_{i} \in \{0,1\},\quad \forall i\in N.&\label{eqCKM:binary constraint}
\end{alignat}
Replacing constraints (\ref{eqCKM:binary constraint}) by $0 \leq y_{i}\leq 1, \forall i\in N$, we obtain the LP-relaxation of CKM.

\subsection{Related Work and Our Results}
The $k$-median problem  is a classical NP-hard problem in computer science and operations research, and has a wide variety of applications in clustering and data mining \cite{BradleyFM,JainD}. The uncapacitated $k$-median problem was studied extensively \cite{ArcherRS,AryaGKMMP,ByrkaPRST,CharikarG,CharikarGTS,JainMS,JainV,LiS}, and the best known approximation algorithm was recently given by Byrka et al. \cite{ByrkaPRST} with approximation ratio $2.611+\epsilon$ by improving the algorithm of Li and Svensson \cite{LiS}.

The capacitated versions of $k$-median problem are much less understood. The above LP-relaxation has an unbounded integrality gap. More precisely, the capacity or the number of opened centers has to be increased by a factor of at least 2, if we try to get an integral solution within a constant factor of the cost of an optimal solution to the LP-relaxation \cite{CharikarGTS}.
All the previous attempts with constant approximation ratios for this problem violate at least one of the two kinds of hard constraints: the capacity constraint and cardinality constraint (at most $k$ centers can be opened), even the local search technique.

For the hard uniform capacity case, by increasing the capacities within a factor of $3$, Charikar et al. \cite{Charikar,CharikarGTS,Guha} presented a $16$-approximation algorithm based on LP-rounding.
This violation ratio of capacities was recently improved to $2+\frac{3}{l-1}$, $l\in \{2,3,4,\dots\}$ by Byrka et al. \cite{ByrkaFRS},
with the corresponding approximation ratio of  $32 l^2+28l+7$. In addition, Korupolu et al. \cite{KorupoluPR} proposed a $(1+{5}/{\epsilon})$-approximation algorithm while opening at most $(5+\epsilon)k$ centers, and a $(1+{\epsilon})$-approximation algorithm while opening at most $(5+5/\epsilon)k$ centers based on a local search technique.

For soft non-uniform capacities, Chuzhoy and Rabani \cite{ChuzhoyR} presented a $40$-approximation algorithm while violating the capacities within a factor of $50$ based on primal-dual and Lagrangian relaxation methods. For hard non-uniform capacities, Gijswijt and Li \cite{GijswijtL} gave a ($7+\epsilon$)-approximation algorithm while opening at most $2k$ centers.

In this paper, we improve the algorithm of Charikar et al. \cite{CharikarGTS} to reduce its violation ratio of capacities from 3 to $2+\frac{2}{\alpha}, \alpha\geq 4$ and get an $(6+10\alpha)$-approximation algorithm for the hard uniform capacitated $k$-median problem, which improves the previous best approximation ratio for any violation ratio of capacities in $(2,3)$. The approximation ratios we obtain for violation ratio of 2.1, 2.3, 2.5, 2.75 and 3 (for instance) are summarized in the following table.
\begin{center}
  \begin{tabular}{ c | c | c | c | c | c }
    violation ratio of capacities & 2.1 & 2.3  & 2.5 & 2.75 & 3  \\ \hline
    previous best & 31627  & 4187  & 1771 & 947 & 16  \\ \hline
    our algorithm & 206  & 72.67  & 46 & 46  & 46  \\
  \end{tabular}
\end{center}
Note that with increasing the capacities by a factor of at least 3, the best approximation ration is still due to Charikar et al. \cite{CharikarGTS}.

Additionally, for metric facility location problems there is a slightly different model for the capacitated $k$-median \cite{ByrkaFRS,GijswijtL},
in which we are given a set $F$ of facilities and a set $D$ of clients.
Each facility has a capacity $M$. Each client $j\in D$ has a demand $d_j$ that has to be served by facilities.
Note that the capacity of each client is 0. This is different from our model, in which each location has a capacity $M$.
We show that our algorithm can be easily extended to solve this model with increasing the approximation ratio
by a factor at most $2+\frac{1}{6+10\alpha}$ (the violation ratio of capacities is the same, see Appendix \ref{appendix:extent to solve another model} for details).

\subsection{The Main Idea Behind Our Algorithm}
Based on an optimal solution to the LP-relaxation, Charikar et al. \cite{CharikarGTS} construct a $\{\frac{1}{2},1\}$-solution
$(x,y)$ in which $y_i\in \{0,\frac{1}{2},1\},\forall i\in N;$ $\sum_{j\in N} x_{ij}d_j\leq M$, if $y_i=\frac{1}{2};$ and
$\sum_{j\in N} x_{ij}d_j\leq 2M$, if $y_i=1.$
Note that $\sum_{j\in N} x_{ij}d_j\leq M y_i$ could be violated in this solution.

First, they directly build a center at location $i$ with $y_i=1$. Then, they construct a collection of rooted stars spanning the locations $i\in N$ with $y_i=\frac{1}{2}.$
By a star by star rounding procedure, exactly half of the locations with fractional opening value $\frac{1}{2}$ are chosen as centers, and reassign the demand served by other locations (not chosen as centers) to the centers.
In the worst case, the capacity of the root of some star has to be increased by factor 3 to satisfy the capacity constraint.
Take Fig.\ref{example(CKM)} as an example. The star $Q_t$, rooted at $t$, has two children $j_1$ and $j_2$ with $y_t=y_{j_1}=y_{j_2}=\frac{1}{2}$.
In the worst case of Charikar et al. algorithm, we are allowed to build at most $\lfloor y_t+y_{j_1}+y_{j_2} \rfloor$ centers, i.e., 1 center.
Suppose we build a center at the root $t$, and reassign the demand served by ${j_1}$ and ${j_2}$ to $t$.
So the capacity of $t$ has to be increased by factor 3 to satisfy the capacity constraint,
as $\sum_{j\in N} x_{ij}d_j\leq M$ for $i=t,j_1,j_2$.

\begin{figure}[!htb]
\centering
\includegraphics[totalheight=1.0in]{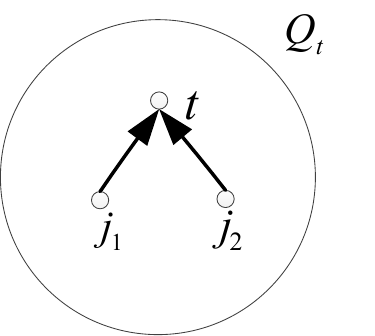}
\caption{\textsl{A star $Q_t$.}}
\label{example(CKM)}
\end{figure}

We generalize the algorithm of Charikar et al. to improve its violation ration from 3 to $2+\epsilon$.
The key idea behind our algorithm relies on two observations.
One is that if we can obtain a $\{1-\frac{1}{\delta},1\}$-solution,
then 2 centers can be built for the above example in the worst case by setting $\delta\geq 3,$ as then $\lfloor y_t+y_{j_1}+y_{j_2} \rfloor \geq \lfloor \frac{2}{3}+\frac{2}{3}+\frac{2}{3}\rfloor= 2$.
Consequently, we only need to blow up the capacity of location $t$ by factor 2 instead of 3, by building centers at $t$ and $j_2$, and assigning the demand served by ${j_1}$ to $t$.
However, this example only shows one kind of stars. To make sure the violation ratio can be improved for all kinds of stars, we construct a $\{(\frac{\alpha-2}{\alpha},\frac{\alpha-1}{\alpha}], [1,2)\}$-solution $(x,y)$ such that \begin{itemize}
\item[1.] for each $i\in N$, $\frac{\alpha-2}{\alpha}< {y}_i\leq \frac{\alpha-1}{\alpha}$, or $1\leq {y}_i< 2$, or ${y}_i=0$; and
$|\{i\in N \mid \frac{\alpha-2}{\alpha}< {y}_i < \frac{\alpha-1}{\alpha}\}|\leq 1;$
\item[2.] if $\frac{\alpha-2}{\alpha}< {y}_i\leq \frac{\alpha-1}{\alpha}$, then  $\sum_{j\in N} d_j x_{ij}\leq M;$
\item[3.] if $1\leq {y}_i< 2$, then  $\sum_{j\in N} d_j x_{ij}\leq M {y}_i.$
\end{itemize}

Another one is that constraints $y_i\leq 1, \forall i\in N$ hold in each step of the algorithm by Charikar et al. That is, they round $y_i>1$ to be 1 for each $i\in N$ in each step.
This is a quite natural operation since we consider the hard capacitated case,
i.e., at most one center can be opened at any location.
However, we observe that after obtaining an optimal solution to the LP-relaxation, it is sufficient to make sure constraints $y_i\leq 1, \forall i\in N$ hold in our last step. For all other steps (except last step), this rounding can be avoided by relaxing the constraint $y_i\leq 1$ to $y_i< 2$. We use an example to show the profit we can gain from avoiding this rounding. Suppose we have a star $Q_t$ rooted at $t$ with one child $j_1$. Moreover, $y_t=1.9$ and $y_{j_1}=0.5$. Then, in the worst case, we can build $\lfloor y_t+y_{j_1}\rfloor= 2$ centers. We open $t$ and $j_1$. Consequently, we only need to increase the capacity of $t$ by factor 1.9 (note that if $1\leq {y}_i< 2$, then  $\sum_{j\in N} d_j x_{ij}\leq M {y}_i$ for our $\{(\frac{\alpha-2}{\alpha},\frac{\alpha-1}{\alpha}], [1,2)\}$-solution). However, if we round 1.9 to 1, we obtain a star $Q_t$ with $y_t=1$ and $y_{j_1}=0.5$. Then, in the worst case, only 1 center can be built as $\lfloor y_t+y_{j_1}\rfloor= 1$. Without loss of generality, suppose we build a center at $t$, and assign the demand served by ${j_1}$ to $t$. Then, we need to increase the capacity of $t$ by factor 2.9.

\section{An Improved Approximation Algorithm}\label{sec:approximation algorithm for CKM}
We consider $y_i$ as the \textit{opening value} of location $i$. If $y_i\in (0,1)$, we say that location $i$ is \textit{fractionally opened} (as a center).
From now on, let $(x,y)$ denote an optimal solution to the LP-relaxation with total cost $C_{LP}$.
For each $j\in N$, define $C_j=\sum_{i\in N} c_{ij}x_{ij}.$ Note that $C_{LP}=\sum_{j\in N} d_j C_j.$ The outline of our algorithm is similar to \cite{CharikarGTS}.

Step 1. We partition locations to a collection of clusters.
The total opening value of each cluster is at least $\frac{\alpha-1}{\alpha}, \alpha\geq 4.$

Step 2. For each cluster, we integrate the nearby opened locations to obtain a
$[\frac{\alpha-1}{\alpha},2)$-solution $(x',y')$ to the LP-relaxation, which satisfies the relaxing constraints $0\leq y'_i< 2$ instead of $0\leq y'_i\leq 1$ for each $i\in N$.

Step 3. We redistribute the opening values among locations with $y'_i\in [\frac{\alpha-1}{\alpha},1)$ to obtain a $\{(\frac{\alpha-2}{\alpha},\frac{\alpha-1}{\alpha}], [1,2)\}$-solution $(x',\hat{y})$, which satisfies the relaxing constraints $\sum_{j\in N}{d_j x'_{ij}} \leq  M $ if $\hat{y}_i\in (0,1)$, $\sum_{j\in N}{d_j x'_{ij}} \leq  M \hat{y}_i$ otherwise, instead of $\sum_{j\in N}{d_j x'_{ij}} \leq  M \hat{y}_i$ for each $i\in N$.

Step 4. We round the $\{(\frac{\alpha-2}{\alpha},\frac{\alpha-1}{\alpha}], [1,2)\}$-solution to be an integral solution with increasing the capacities by a factor of $2+\frac{2}{\alpha}$.

\subsection{Step 1: Clustering}
In this step, we will partition locations into clusters, and for each cluster select a single location as the \textit{core} of this cluster,
such that each location in the cluster is not far to its cluster core and the cores are sufficiently far to each other.

Let $N'$ be the collection of all cluster cores. Let $N'(j)$ denote the closest cluster core to $j$ in $N'$.
For each $l\in N'$, let $M_l$ denote the cluster whose core is $l$,
and define $Z_l=\sum_{j\in M_l} {y_j}$ be the total opening value of all locations in cluster $M_l$.
\begin{definition}We call a cluster $M_l$ \textit{terminal} if $Z_l\geq 1$, \textit{non-terminal} if $0<Z_l<1$.\end{definition}

Let $n=|N|.$ The clustering is done as follows.\\
\begin{procedure}[H]
\caption{1. Clustering()}
1. order all locations in nondecreasing order of ${C_j}$, (without loss of generality, assume ${C_1}\leq \cdots \leq {C_n}$)\;
2. set $N':=\emptyset$\;
3. \For{$j=1$ to $n$}
 {
    find a location $l\in N'$ such that $c_{lj}\leq 2\alpha C_j,$  where $\alpha\geq 4$\;
    \If{no such location is found}
    {
     choose $j$ as a cluster core, i.e., set $N':=N'\cup \{j\}$\;
    }
 }
4. set $M_l:=\emptyset$, $\forall l\in N'$\;
5. \For{$j=1$ to $n$}
 {
    \If{$j$ is closer to cluster core $l\in N'$ than all other cluster cores (break ties arbitrarily)}
    {
     add location $j$ to cluster $M_l$. (i.e., set $M_l:=\{j\in N \mid N'(j)=l\}$.)
    }
 }
\end{procedure}
After this step, the following properties hold ($\alpha\geq 4$):
\vspace{5pt}

[\textbf{1a}]. $\forall j\in M_l, l\in N'$, $c_{lj}\leq 2\alpha {C_j}$;

[\textbf{1b}]. $\forall l,l'\in N'$ and $l\neq l'$, $c_{ll'}> 2\alpha\max\{C_l, C_{l'}\}$;

[\textbf{1c}]. $\forall l\in N', Z_l=\sum_{j\in M_l} {y_j}\geq \frac{\alpha-1}{\alpha}$;

[\textbf{1d}]. $\bigcup_{l\in N'} M_l=N;$ and $M_l\bigcap M_{l'}=\emptyset, \forall l,l'\in N'$ and $l\neq l'$.

\vspace{5pt}
We can easily get property \textbf{1a}, \textbf{1b} and \textbf{1d} from the Procedure Clustering.
\begin{lemma}\label{lemma:property 1c}
\textup{(property \textbf{1c})} $\forall l\in N', Z_l\geq \frac{\alpha-1}{\alpha}$. \textup{(See Appendix \ref{appendix:property 1c} for the details of proof.)}
\end{lemma}
\noindent We give a brief idea. First, we show $i\in M_l$ if $c_{il}\leq \alpha C_l$.
Then, note that $\sum_{i\in N : c_{il}>\alpha C_l} x_{il}<\frac{1}{\alpha}$, otherwise $\sum_{i\in N}{c_{il} x_{il}}>C_l$, a contradiction. So, $Z_l\geq \sum_{i\in N : c_{il}\leq \alpha C_l} y_{i}\geq \sum_{i\in N : c_{il}\leq \alpha C_l} x_{il} \geq \frac{\alpha-1}{\alpha}.$

\subsection{Step 2: Obtaining a $[\frac{\alpha-1}{\alpha},2)$-solution}
We will get rid of locations with relatively small fractional opening value in this step, by constructing a $[\frac{\alpha-1}{\alpha},2)$-solution $(x',y')$ in which $y'_i=0$ or $\frac{\alpha-1}{\alpha}\leq y'_i<2$, $\forall i\in N$.
For each cluster $M_l$, we transfer the amount of locations (their opening values and the demands served by these locations) far away from the cluster core $l$ to locations closer to $l$.




In this step, initially set $y'_i=y_i, x'_{ij}=x_{ij} \forall i,j\in N$.
Then, we consider clusters one by one. For each cluster $M_l, l\in N'$, order locations in $M_l$ in nondecreasing value of $c_{lj}, j\in M_l$.
Without loss of generality, assume we get an order $j_1, \cdots, j_u$ (Note that $j_1=l$).
If we decide to move the amount of location $j_b$ to $j_a$ ($1\leq a<b\leq u$), then perform the following transfer operations \cite{Charikar,Guha}:\\
\begin{procedure}[H]
\caption{2. Move($j_a$,$j_b$)}
 {
  1. let $\delta=\min \{1-y'_{j_a},y'_{j_b}\}$\;
  2. for all $j\in N$, set $x'_{j_a j}:=x'_{j_a j}+\frac{\delta}{y'_{j_b}} x'_{j_b j}$,
  $x'_{j_b j}:=x'_{j_bj}-\frac{\delta}{y'_{j_b}} x'_{j_b j}$\;
  3. set $y'_{j_a}:=y'_{j_a}+\delta$,$y'_{j_b}:=y'_{j_b}-\delta$\;
 }
\end{procedure}

\begin{lemma}\label{lemma:Move}
After Procedure Move($j_a$,$j_b$), we still have
\begin{itemize}
\item[(1)] $\sum_{j\in M_l} y'_j = \sum_{j\in M_l} y_j,$ for each $l\in N'$;
\item[(2)] for each $j\in N$, $\sum_{i\in N}{x'_{ij}}=1$;
\item[(3)] $\sum_{j\in N} d_j x'_{ij}\leq M y'_i$, for each $i\in N$. \textup{(See Appendix \ref{appendix:Move} for the proof.)}
\end{itemize}
\end{lemma}

We use the following procedure to decide whether we move the amount of location $j_b$ to $j_a$.\\
\begin{procedure}[H]
\caption{3. Concentrate($M_l$)}
 \While{there exists a location in $M_l$ with fractional opening value}
 {
  1. let $j_a$ be the first location in the sequence $j_1,\cdots,j_u$ such that $0\leq y'_{j_a}<1$\;
  2. let $j_b$ be the first location in the sequence $j_{a+1},\cdots,j_u$ such that $0<y'_{j_b}\leq 1$\;
  3. \If{$j_a$ and $j_b$ both exist}
  {
   execute procedure Move($j_a$,$j_b$) to move the amount of $j_b$ to $j_a$\;
  }
  4. \If{$j_a$ exists but $j_b$ does not exist}
  {
    \If{$M_l$ is a terminal cluster,i.e.,$a\geq 2$}
    {
     set $y'_{j_{a-1}}:=y'_{j_{a-1}}+y'_{j_{a}}, y'_{j_a}:=0$\;
     for each $j\in N$, set $x'_{j_{a-1} j}:=x'_{j_{a-1} j}+x'_{j_{a} j}, x'_{j_{a} j}:=0$\;
    }
   terminate.
  }
 }
\end{procedure}
\begin{lemma}\label{lemma:last step of concentrating a terminal cluster}
If in Procedure 3 $j_a$ exists but $j_b$ does not exist, and $M_l$ is a terminal cluster, then $a\geq 2$ and $y'_{j_{a-1}}=1$.
\end{lemma}
\begin{proof}
Since $M_l$ is a terminal cluster, we have $Z_l\geq 1$.
Moreover, we know $y'_{j_t}=1$ for each $t<a$ and $y'_{j_s}=0$ for each $s>a$, as $j_b$ does not exist.
Thus, $a\geq 2$. Otherwise, $Z_l<1$, a contradiction.
\end{proof}

\begin{lemma}\label{lemma:Concentrate}
 After this step, we have the following properties

 \textup{[\textbf{2a}]}. for all $i\in N$, $\frac{\alpha-1}{\alpha}\leq y'_i<2$ or $y'_i=0$; and $\sum_{j\in N} d_j x'_{ij}\leq M y'_i$;

 \textup{[\textbf{2b}]}.  $\sum_{i\in N} y'_i = \sum_{i\in N} y_i\leq k;$

 \textup{[\textbf{2c}]}.  $x'_{ij} \leq y'_i, \forall i, j\in N.$
\end{lemma}
\begin{proof}
Property \textbf{2a.} If $M_l$ is a non-terminal cluster, i.e., $0<Z_l<1$,
then we will move the amount of each location in $M_l$ to its core $l$ according to Procedure 3.
Consequently, we obtain $\frac{\alpha-1}{\alpha}\leq y'_l=Z_l<1$ (property \textbf{1c}) and $y'_j=0, \forall j\in M_l-\{l\}.$

If $M_l$ is a terminal cluster, i.e., $Z_l\geq 1$, then according to Lemma \ref{lemma:last step of concentrating a terminal cluster}
we get $y'_{j_t}=1$ for each $t<a$ and $y'_{j_s}=0$ for each $s>a$ if $j_a$ exists and $j_b$ does not exist.
Then, we move the amount of $y'_{j_a}$ to $y'_{j_{a-1}}$. So, $1\leq y'_{j_{a-1}}<2$ as $0\leq y'_{j_a}<1$.
Note that if $j_a$ does not exist, then we know $y'_j=1$ for each $j\in M_l$.

Thus, for all $i\in N$, $\frac{\alpha-1}{\alpha}\leq y'_i<2$ or $y'_i=0$. $\sum_{j\in N} d_j x'_{ij}\leq M y'_i, \forall i\in N$ hold by Lemma \ref{lemma:Move} (Note that it is easy to check these inequalities still hold after the step 4 in Procedure 3).

Property \textbf{2b.} This directly follows by Lemma \ref{lemma:Move}(1).

Property \textbf{2c.} We give a brief idea here and see Appendix \ref{appendix:Concentrate} for details.
Observe that for each $j\in N$, we always set $x'_{ij}:=0$ if $y'_i$ is already set to be 0.
For each non-terminal cluster, only the core has a positive opening value after this step.
And in the procedure the opening value of core is always increased by a bigger amount than the increasing of the fraction of the demand served by the core.
For a terminal cluster, each location $i$ in the cluster has $y'_i=0$ or $y'_i\geq 1$ after this step.
Note that for each location $i\in N$ with $y'_i\geq 1$, $x'_{ij}\leq y'_i$ holds for each $j\in N$ as $x'_{ij}\leq 1.$
\end{proof}

Since each location is not far away from its cluster core, these transfer operations would not increase too much extra cost.
\begin{lemma}\label{lemma:extra cost of step 2}
 (1). Let $M_l$ be a non-terminal cluster. The demand of location $j$ originally served by $j_b$($j_b\in M_l$) must be served by core $l$ after the procedure. And we have $c_{l j}\leq 2c_{j_b j}+ 2\alpha C_j.$

 (2). Let $M_l$ be a terminal cluster. If we move the demand of location $j$ served by $j_b$ to $j_a$ ($j_a, j_b\in M_l$, $a<b$), we have $c_{j_a j}\leq 3c_{j_b j}+4\alpha C_j.$ \textup{(See Appendix \ref{appendix:extra cost of step 2} for the proof.)}
\end{lemma}

Let $N_1=\{i\in N \mid y'_i\geq 1\}$ be the collection of locations with the opening value at least 1. Let $N_2=\{i\in N \mid y'_i\in [\frac{\alpha-1}{\alpha},1)\}$ be the collection of locations with fractional opening value in $[\frac{\alpha-1}{\alpha},1)$. Note that $N_2$ can also be written as $\{i\in N' \mid Z_i\in [\frac{\alpha-1}{\alpha},1)\}.$ That is, $N_2$ is the collection of non-terminal cluster cores. Moreover, we have $N_1\cup N_2 \supseteq N'.$

\begin{lemma}\label{lemma:easy case}
If $|N_2|-1< \sum_{i\in N_2} y'_i$, we can get an integer solution with increasing the capacity by factor 2, by opening all locations in $N_1\cup N_2$ as centers. The total cost of the obtained solution can be bounded by $(3+4\alpha)C_{LP}$.
\end{lemma}
\begin{proof}
If $|N_2|-1< \sum_{i\in N_2} y'_i$, then $|N_2|= \lceil \sum_{i\in N_2} y'_i \rceil$ as $y'_i<1$ for each $i\in N_2$. Additionally, since $\sum_{i\in N_1} y'_i\leq k-\sum_{i\in N_2} y'_i$ (by property \textbf{2b}) and $y'_i\geq 1$ for each $i\in N_1$, we have $|N_1|\leq \lfloor k-\sum_{i\in N_2} y'_i \rfloor.$

Thus, if we only open locations in $N_1\cup N_2$, then we open at most $k$ centers as $\lceil \sum_{i\in N_2} y'_i \rceil+\lfloor k-\sum_{i\in N_2} y'_i \rfloor=k$.

Since $y'_i=0$ for each $i\notin N_1\cup N_2$, we have $\sum_{i\in N_1\cup N_2} x'_{ij}=1, \forall j\in N$ by Lemma \ref{lemma:Move}(2) and property \textbf{2c.} That is, $\sum_{i\in N_1\cup N_2} d_j x'_{ij}=d_j$ for each $j\in N.$ Thus, the demand of each $j\in N$ can be satisfied by assigning $d_j x'_{ij}$ to $i\in N_1\cup N_2$.

By Lemma \ref{lemma:extra cost of step 2}, it is easy to see the total cost of the obtained solution can be bounded by $(3+4\alpha)C_{LP}$. By Lemma \ref{lemma:Concentrate}, we know for all $i\in N$, $\frac{\alpha-1}{\alpha}\leq y'_i<2$ or $y'_i=0$; and $\sum_{j\in N} d_j x'_{ij}\leq M y'_i$. So, we increase the capacity by at most a factor of 2.
\end{proof}

From now on, we only consider the following case.
\begin{assumption}\label{ass:nontrivial case}
$\sum_{i\in N_2} y'_i\leq |N_2|-1.$
\end{assumption}

\begin{definition}\label{def:new demands}
We define new demands $d'$ as follows. For $i\in N$, set $d'_i:=\sum_{j\in N}{d_j x'_{ij}}.$ (Note that $d'_i=0$ for each $i\in N-(N_1\cup N_2).$)
\end{definition}

\subsection{Step 3: Obtaining a $\{(\frac{\alpha-2}{\alpha},\frac{\alpha-1}{\alpha}], [1,2)\}$-solution}
For each $i\in N_2$, let $s(i)$ be the nearest location to $i$ in $(N_1\cup N_2)-\{i\}$ (break ties arbitrarily).
Let $Y=\sum_{i\in N_2} {y'_i}$.
Note that we only consider the case: $Y\leq |N_2|-1$ by Assumption \ref{ass:nontrivial case}.
After this step we will obtain a solution $(x',\hat{y})$ with $\frac{\alpha-2}{\alpha}< \hat{y}_i\leq \frac{\alpha-1}{\alpha}$, or $1\leq \hat{y}_i< 2,$ or $\hat{y}_i=0$ for each $i\in N$.

In this step, initially we order all locations in $N_2$ in nondecreasing order of $d'_i c_{s(i) i}$.
Without loss of generality, suppose we get an order $i_1,\cdots,i_v$.

Then, for each $i\in N-N_2$, set $\hat{y}_i:=y'_i$. For each $i\in N_2$, set $\hat{y}_i:= \frac{\alpha-1}{\alpha}$. Let $Y':=Y-\sum_{i\in N_2} {\hat{y}_i}$. Then, perform the following operations:\\
\begin{procedure}[H]
\caption{4. Determine new opening values for $N_2$($Y\leq |N_2|-1$)}
 \For{$r=v$ to 1}
 {
  \If{$Y'=0$}
  {
   terminate\;
  }
  \If{$Y'>0$ and $Y'+\hat{y}_{i_r}<1$}
  {
   set $\hat{y}_{i_1}:=\hat{y}_{i_1}-(1-Y'-\hat{y}_{i_r}), \hat{y}_{i_r}:=1$\;
   terminate\;
  }
  \If{$Y'>0$ and $Y'+\hat{y}_{i_r}\geq 1$}
  {
   set $\hat{y}_{i_r}:=1$ and update $Y':=Y-\sum_{i\in N_2}{\hat{y}_i}$\;
  }
 }
\end{procedure}
\begin{remark}\label{remark:step 3 terminates}
The Procedure 4 terminates at $r>1$. If the procedure terminates at $r=1$, then we get $Y=\sum_{t=1}^v y'_{i_t}> |N_2|-1,$ a contradiction. See Appendix \ref{appendix:step 3 terminates} for the details of argument.
\end{remark}

\begin{lemma}\label{lemma:Determine new opening values}
After the above procedure, we have the following properties

 \textup{[\textbf{3a}]}. for all $i\in N$, $\frac{\alpha-2}{\alpha}< \hat{y}_i\leq \frac{\alpha-1}{\alpha}$, or $1\leq \hat{y}_i< 2,$ or $\hat{y}_i=0;$ and only $\hat{y}_{i_1}$ can be in $(\frac{\alpha-2}{\alpha},\frac{\alpha-1}{\alpha})$, i.e., $|\{i\in N \mid \frac{\alpha-2}{\alpha}< \hat{y}_i < \frac{\alpha-1}{\alpha}\}|\leq 1;$

 \textup{[\textbf{3b}]}. for any location $i\in N$, if $\frac{\alpha-2}{\alpha}< \hat{y}_i\leq \frac{\alpha-1}{\alpha}$, then  $d'_i=\sum_{j\in N} d_j x'_{ij}\leq M;$

 \textup{[\textbf{3c}]}. for any location $i\in N$, if $1\leq \hat{y}_i< 2$, then  $d'_i=\sum_{j\in N} d_j x'_{ij}\leq M \hat{y}_i;$

 \textup{[\textbf{3d}]}. $\sum_{i\in N_2}{\hat{y}_i}=\sum_{i\in N_2}{y'_i};$ $\sum_{i\in N}{\hat{y}_i}=\sum_{i\in N}{y'_i}\leq k;$

 \textup{[\textbf{3e}]}. $\sum_{i\in N_2}(1-\hat{y}_i)d'_i c_{s(i) i} \leq \sum_{i\in N_2}(1-y'_i)d'_i c_{s(i) i}.$
\end{lemma}
\begin{proof}
We give some brief ideas here. See Appendix \ref{appendix:Determine new opening values} for the details.

Property \textbf{3a}. For each location $i\in N-N_2$, we set $\hat{y}_i:=y'_i$. So, $1\leq \hat{y}_i< 2$ for each $i\in N_1$; $\hat{y}_i=0$ for each $i\in N-(N_1\cup N_2)$.

For each location $i\in N_2$, initially we set $\hat{y}_{i}:=\frac{\alpha-1}{\alpha}$. In the Procedure 4, only $\hat{y}_{i_1}$ could be decreased by a number in $(0,\frac{1}{\alpha}).$ The opening value of other location in $N_2$ remains the same or is set to be $1$.

Property \textbf{3b, 3C.} Notice that if for location $i$ we have $\frac{\alpha-2}{\alpha}< \hat{y}_i\leq \frac{\alpha-1}{\alpha}$ after the procedure, then we know $\frac{\alpha-1}{\alpha}\leq y'_i<1.$
And if $1\leq  \hat{y}_i<2$ for location $i$ after the procedure, then we have $y'_i\leq \hat{y}_i$.

We make no change on $x'$. Then, combining with property \textbf{2a}, we have if $\frac{\alpha-2}{\alpha}< \hat{y}_i\leq \frac{\alpha-1}{\alpha}$, then $\sum_{j\in N} d_j x'_{ij}\leq M y'_i< M.$ If $1\leq \hat{y}_i< 2$, then $\sum_{j\in N} d_j x'_{ij}\leq M y'_i\leq  M \hat{y}_i.$

Property \textbf{3d}. We move the opening value from one location to the other locations. We do not change the total opening value.
So, $\sum_{i\in N_2}{\hat{y}_i}=\sum_{i\in N_2}{y'_i}$ holds after the above process. Moreover, we set $\hat{y}_i:=y'_i$ for each $i\in N-N_2$. Thus, we also have $\sum_{i\in N}{\hat{y}_i}=\sum_{i\in N}{y'_i}\leq k.$

Property \textbf{3e}. We always transfer the opening value from
$i_a$ to $i_b$, where $a<b$ and $d'_{i_b} c_{s({i_b}) {i_b}}\geq d'_{i_a} c_{s({i_a}) {i_a}}$.
Therefore, $\sum_{i\in N_2} \hat{y}_i d'_i c_{s(i) i} \geq \sum_{i\in N_2} y'_i d'_i c_{s(i) i}.$ Then, we have $\sum_{i\in N_2}(1-\hat{y}_i)d'_i c_{s(i) i} \leq \sum_{i\in N_2}(1-y'_i)d'_i c_{s(i) i}.$
\end{proof}

\subsection{Step 4: Rounding to an Integral Solution}
Let $\hat{N}_1=\{i\in N \mid 2> \hat{y}_i\geq 1\}$ be the set of locations with opening value greater than or equal to 1.
Let $\hat{N}_2=\{i\in N \mid \frac{\alpha-2}{\alpha}< \hat{y}_i\leq \frac{\alpha-1}{\alpha}\}$ be the set of location with fractional opening value strictly less than 1. Let $L_1=|\hat{N}_1|$. Note that $N_1\cup N_2=\hat{N}_1\cup \hat{N}_2,$ and $\hat{N}_2\subseteq N_2.$

In this step, we aim to construct an integral solution $(\bar{x},\bar{y})$ with $\sum_{j\in N} \bar{x}_{ij}d'_j\leq (2+\frac{2}{\alpha})M \bar{y}_i$ for each $i\in N$.
If location $j$ is opened as a center, we serve the demand $d'_j$ of location $j$ by itself. That is, set $\bar{x}_{jj}:=1, \bar{x}_{ij}:=0$ for each $i\neq j, i\in N.$
And we build a center at location $i$ if $1\leq \hat{y}_i< 2$, i.e., set $\bar{y}_{i}:=1$ for each $i\in \hat{N}_1$.
For $\hat{N}_2$, we will open at most $k-L_1$ locations as centers.
If a center is not opened at location $j\in \hat{N}_2$, we assign the demand $d'_j$ of $j$ to another opened center $i$, i.e., set $\bar{x}_{ij}:=1$.
Now we start to show the details of this step.

Initially, for each $i,j\in N$ set $\bar{x}_{ij}:=0;$ and $\bar{y}_{i}:=0$.
Then, we construct a collection of rooted trees spanning the locations in $\hat{N}_2$ as in \cite{CharikarGTS}.
Recall that $s(i)$ is the closest location to $i$ in $(\hat{N}_1\cup \hat{N}_2)-\{i\}$ ($N_1\cup N_2=\hat{N}_1\cup \hat{N}_2$) for each $i\in N_2$.
We draw a directed edge from $i$ to $s(i)$ if $i\in \hat{N}_2$.
The cycles can be eliminated by the following way.
For each cycle, we take any location in this cycle as a root and delete the edge from this root to other location.
If there is a directed edge from $i$ to $s(i)$ finally, we consider $s(i)$ as the parent of $i$. Then, we get a desired collection of rooted trees.

Next, we decompose each tree into a collection of rooted stars by the following procedure.
\begin{remark}\label{remark:star}
In each rooted star, all the children of the root have a fractional opening value.
If the root of a star is a fractionally opened location, then the root has at least one child.
\end{remark}
\begin{procedure}[H]
\caption{5. Decompose a tree $T$ to stars()}
\While{there are at least two nodes in $T$}
 {
  choose a leaf node $i$ with biggest number of edges on the path from $i$ to the root\;
  consider the subtree rooted at $s(i)$ as a rooted star, and remove this subtree\;
 }
\If{only one node $i$ is left and $0<\hat{y}_i<1$}
 {
  add $i$ to the star rooted at $s(i)$ as a child of $s(i)$\;
 }
\end{procedure}

\begin{definition}
An even star is a star with even number of children. An odd star is a star with odd number of children.
\end{definition}

Let $Q_t$ denote the star rooted at location $t$. By abuse of notation, we also use $Q_t$ to denote the collection of locations in the star rooted at $t$. Let $R_t=\sum_{i\in Q_t} \hat{y}_i$ be the total opening value in $Q_t$.

\begin{lemma}\label{lemma:total fractional opening value}
(1) If a star $Q_t$ has even number of fractionally opened locations, i.e., $|Q_t\cap \hat{N}_2|=2q$ is an even number ($q\in \mathbb{Z}^+$), then the total opening value of these fractionally opened locations is greater than $q$, i.e., $\sum_{i\in Q_t\cap \hat{N}_2} \hat{y}_i > q.$

(2) If $|Q_t\cap \hat{N}_2|=2q+1$ is an odd number and $q\in \mathbb{Z}^+$, then $\sum_{i\in Q_t\cap \hat{N}_2} \hat{y}_i> q+1.$
\end{lemma}
The proof for the above lemma is given in Appendix \ref{appendix:total fractional opening value}.
The proof is based on the fact that at most one location $i\in Q_t\cap \hat{N}_2$ has the opening value
in $(\frac{\alpha-2}{\alpha}, \frac{\alpha-1}{\alpha})$ and the opening value of other location
in $Q_t\cap \hat{N}_2$ is exactly $\frac{\alpha-1}{\alpha}$ (property \textbf{3a}).

We build a center at each location $i\in \hat{N}_1- \bigcup_t Q_t$ (locations are in $\hat{N}_1$, but not in any star), i.e., set $\bar{y}_{i}:=1$ and $\bar{x}_{ii}:=1$.
For each kind of star $Q_t$, we define operations to make sure at most $\lfloor R_t \rfloor$ locations in $Q_t$ are selected to be centers.

\textbf{1. An even star rooted at location $t$ with $1\leq \hat{y}_t<2$.}
Let $i_1,\cdots,i_{2q}$ be a sequence of all its children in nondecreasing order of distance from $t$.
We build centers at location $t, i_1,i_3,\cdots,i_{2q-1}$, and serve the demand $d'_{i_{2r}}$ of $i_{2r}$ by opened location $i_{2r-1}$, i.e.,
\begin{alignat}{2}
\text{set } &\bar{y}_{t}:=1; \quad \bar{y}_{i_{2r-1}}:=1, \bar{y}_{i_{2r}}:=0, r=1,\cdots,q; \nonumber\\
\text{set } &\bar{x}_{tt}:=1; \quad \bar{x}_{i_{2r-1} i_{2r-1}}:=1, \bar{x}_{i_{2r-1} i_{2r}}:=1, r=1,\cdots,q.\nonumber
\end{alignat}

\textbf{2. An even star rooted at location $t$ with $\frac{\alpha-2}{\alpha}< \hat{y}_t\leq \frac{\alpha-1}{\alpha}$.}
Let $i_1,\cdots,i_{2q}$ be a sequence of all its children in nondecreasing order of distance from $t$. (Note that $q\geq 1$ by Remark \ref{remark:star}.)
We build centers at location $t,i_2,i_4,\cdots,i_{2q}$, and serve the demand $d'_{i_{2r+1}}$ of $i_{2r+1}$ by opened location $i_{2r}$, serve the demand $d'_{i_{1}}$ of $i_{1}$ by $t$.


\textbf{3. An odd star rooted at location $t$ with $1+\frac{2}{\alpha}\leq \hat{y}_t<2$.}
Let $i_1,\cdots,i_{2q+1}$ be a sequence of all its children in nondecreasing order of distance from $t$.
We open $t,i_1,i_3,\cdots,i_{2q+1}$ as centers, and serve the demand $d'_{i_{2r}}$ of $i_{2r}$ by opened location $i_{2r-1}$.

\textbf{4. An odd star rooted at location $t$ with $\frac{\alpha-2}{\alpha}< \hat{y}_t\leq \frac{\alpha-1}{\alpha}$ or $1\leq \hat{y}_t< 1+\frac{2}{\alpha}$.}
Let $i_1,\cdots,i_{2q+1}$ be a sequence of all its children in nondecreasing order of distance from $t$.
We build centers at location $t,i_2,i_4,\cdots,i_{2q}$, and serve the demand $d'_{i_{2r+1}}$ of $i_{2r+1}$ by opened location $i_{2r}$, serve the demand $d'_{i_{1}}$ of $i_{1}$ by $t$.

Note that $(\bar{x},\bar{y})$ is an integral solution for new demands $d'$.
To get an integral solution for our original demands $d$, we can redistribute the demands $d'$ to their original locations according to Definition \ref{def:new demands}.

By property \textbf{3a}, \textbf{3b} and \textbf{3c}, and Lemma \ref{lemma:total fractional opening value}, we can get the following lemma.
\begin{lemma}\label{lemma:the number of centers of each star}
For each kind of star $Q_t$, we build at most $\lfloor R_t \rfloor$ centers. And for each $i\in N$, we have $\sum_{j\in N}{d'_j \bar{x}_{ij}}\leq (2+\frac{2}{\alpha})M \bar{y}_i$. \textup{(See Appendix \ref{appendix:the number of centers of each star} for the details of proof.)}
\end{lemma}

\begin{lemma}\label{lemma:blowing up factor for capacity}
We build at most $k$ centers, and increase capacities by factor $2+\frac{2}{\alpha}.$
\end{lemma}

\begin{proof}
Suppose we get stars $Q_1,\cdots,Q_t$ by decomposing all the trees.
Then by property \textbf{3d}, we know $\sum_{r=1}^t R_r+\sum_{i\in \hat{N}_1-\bigcup_{r=1}^t Q_r} \hat{y}_i\leq k$.
Moreover, we build at most $\sum_{r=1}^t \lfloor R_r\rfloor+\sum_{i\in \hat{N}_1-\bigcup_{r=1}^t Q_r} \lfloor\hat{y}_i\rfloor$ centers by Lemma \ref{lemma:the number of centers of each star} and the operation for locations that are in $\hat{N}_1$ but not in any star.
Consequently, we build at most $k$ centers. Again, by Lemma \ref{lemma:the number of centers of each star} we increase the capacity by at most a factor of $2+\frac{2}{\alpha}$ to satisfy all the demand constraints.
\end{proof}

For each location $i$ in star $Q_t$, let $r(i)\in Q_t$ denote the location that the demand $d'_i$ of $i$ is reassigned to.
Define the cost of star $Q_t$ as $\sum_{i\in Q_t}d'_i c_{r(i)i}.$

\begin{lemma}\label{lemma:extra cost of step 4}
The cost of stars can be bounded by $\sum_{i\in {N}_2}\sum_{j\in N} \sum_{i'\in M_i} d_j(4c_{i'j}x_{i'j}+8\alpha C_j x_{i'j}).$
\end{lemma}
\begin{proof}
We only consider the service cost for demand $d'_i,i\in \hat{N}_2$, as we open a center at each location in $\hat{N}_1$ and serve its demand by itself.
Based on rounding operations in Step 4, property \textbf{3e}, the definition of $d'$, and Concentrate($M_i$),$i\in N_2$, we get an upper bound $\sum_{i\in {N}_2} \sum_{j\in N} \sum_{i'\in M_i} 2\alpha (1-y'_i)d_{j}x_{i'j}c_{s(i)i}$. Then, we show that $2\alpha (1-y'_i)d_{j}x_{i'j}c_{s(i)i}\leq d_{j}(4c_{i'j}x_{i'j}+8\alpha C_j x_{i'j})$ for both cases $N'(j)=i$ and $N'(j)\neq i.$ \textup{(See Appendix \ref{appendix:extra cost of step 4} for the details.)}
\end{proof}
\section{Analysis}
In our algorithm, we reassign the service twice: in Step 2 and Step 4. The cost of reassignment for Step 2 (Step 4) can be bounded by Lemma \ref{lemma:extra cost of step 2} (Lemma \ref{lemma:extra cost of step 4}). Combining these two upper bounds, the total cost can be bounded by
\begin{eqnarray*}
&&\sum_{i\in {N}_2}\sum_{j\in N} \sum_{i'\in M_i} d_j(2c_{i'j}+2\alpha C_j)x_{i'j}+\sum_{i\in N'-{N}_2}\sum_{j\in N}\sum_{i'\in M_i} d_j(3c_{i'j}+4\alpha C_j)x_{i'j}\\
&&+\sum_{i\in {N}_2}\sum_{j\in N} \sum_{i'\in M_i} d_j(4c_{i'j}x_{i'j}+8\alpha C_j x_{i'j})\\
&&\leq \sum_{i\in {N}}\sum_{j\in N} d_j(6c_{ij}+10\alpha C_j)x_{ij}=\sum_{j\in N} d_j(6C_j+10\alpha C_j)=(6+10\alpha)C_{LP}.
\end{eqnarray*}

Then combining with Lemma \ref{lemma:easy case} and \ref{lemma:blowing up factor for capacity}, we can prove the following theorem.
\begin{theorem}\label{lemma:main theorem for CKM}
For any $\alpha\geq 4,$ there is a $(6+10\alpha)$-approximation algorithm for the hard uniform capacitated $k$-median problem with increasing the capacity by factor at most $2+\frac{2}{\alpha}.$
\end{theorem}

\section*{Acknowledgements}

We thank Dion Gijswijt for insightful discussions.

\bibliographystyle{splncs_srt}
\bibliography{KFL}

\newpage
{\large \bf APPENDIX}
\begin{appendix}
\section{Extent Our Algorithm to Solve Another Model}\label{appendix:extent to solve another model}
The capacitated $k$-median location problem can be formulated as the following mixed integer program, where variable $x_{ij}$ indicates the fraction of the demand of client $j$ that is served by facility $i$, and $y_{i}$ indicates if facility $i$ is open. Let $y_i$ take value one if facility $i$ is open and value zero otherwise. We denote this model by CKL.

\begin{alignat}{2}
\min\quad&\sum_{i\in F}\sum_{j\in D}d_j c_{ij} x_{ij}\\
\text{subject to:}\quad&\sum_{i\in F}{x_{ij}} = 1 ,&\forall j\in D,\\
&\sum_{j\in D}{d_j x_{ij}} \leq  M y_i,&\forall i\in F,\\
&\sum_{i\in F}{y_i} \leq  k,\\
&0\leq x_{ij}\leq  y_i, &\forall i\in F,j\in D,\\
&y_{i} \in \{0,1\},&\forall i\in F. \label{eq:binary constraint}
\end{alignat}
Replacing constraints (\ref{eq:binary constraint}) by
\begin{alignat}{2}
0 \leq y_{i}\leq 1, i\in F, \label{binary relaxation}
\end{alignat}
we obtain the LP-relaxation of CKL.

\subsection{The Algorithm}
Let $(x^0,y^0)$ be an optimal solution to the LP-relaxation of CKL. For each facility $i\in F$, define a demand
\[
d^1_i=\sum_{j\in D} d_j x^0_{ij}.
\]

To make use of the algorithm presented in Section \ref{sec:approximation algorithm for CKM}, we set $N:=F$. That is, each location $i\in N$ has a capacity $M$ and demand $d^1_i$. Then, we get an instance of the capacitated $k$-median problem (CKM) considered in Section \ref{sec:approximation algorithm for CKM}. Suppose we get an integral solution $(x^1,y^1)$ for this constructed instance by the algorithm proposed in Section \ref{sec:approximation algorithm for CKM}.

Then, we construct an integral solution $(x^*,y^*)$ for the original instance of CKL by redistributing the demands $d^1_{i'}$ of location (facility) $i'\in N$ back to clients $D$. That is, set $y^*:=y^1;$ and set $x^*_{ij}:=\sum_{i'\in N} (x^1_{ii'} x^0_{i'j}),$ for each $i\in N=F, j\in D.$

\begin{lemma}
$(x^*,y^*)$ is an integral solution for CKL with $\sum_{j\in D} d_j x^*_{ij}\leq (2+\frac{2}{\alpha})M y^*_i$ for each $i\in F,$ where $\alpha\geq 4.$
\end{lemma}
\begin{proof}
First, we show that $\sum_{i\in F} x^*_{ij}=1$ for each $j\in D$. That is, for each client $j\in D$, its demand $d_j$ is satisfied. Note that $N=F.$ For each client $j\in D$, we have
\begin{eqnarray*}
&&\sum_{i\in F} x^*_{ij}=\sum_{i\in N} \sum_{i'\in N} (x^1_{ii'} x^0_{i'j})=\sum_{i'\in N} \sum_{i\in N} (x^1_{ii'} x^0_{i'j})\\
&&=\sum_{i'\in N} (x^0_{i'j}\sum_{i\in N} x^1_{ii'})=\sum_{i'\in N}  x^0_{i'j}=1,
\end{eqnarray*}
where the first equality follows by the definition of $x^*_{ij}$; the fourth equality holds as $(x^1,y^1)$ is an integral solution for the constructed instance, i.e., $\sum_{i\in N} x^1_{ii'}=1,$ for each $i'\in N.$

Second, we show that $\sum_{i\in F}y^*_i\leq k.$ That is, we open at most $k$ facilities. This is trivial by Lemma \ref{lemma:blowing up factor for capacity}.

Third, we show that $x^*_{ij}\leq  y^*_i, \forall i\in F,j\in D.$ This is also trivial, because
\begin{eqnarray*}
&&x^*_{ij}=\sum_{i'\in N} (x^1_{ii'} x^0_{i'j})\leq \sum_{i'\in N} (y^1_{i} x^0_{i'j})=y^1_{i} \sum_{i'\in N} x^0_{i'j}=y^1_{i}=y^*_{i},
\end{eqnarray*}
where the first inequality holds as $(x^1,y^1)$ is an integral solution for the constructed instance;
the third equality holds as $(x^0,y^0)$ is an optimal solution to the LP-relaxation.

Then, we show that $\sum_{j\in D} d_j x^*_{ij}\leq (2+\frac{2}{\alpha})M y^*_i$ for each $i\in F.$ That is, we only violate the capacities by a factor of $2+\frac{2}{\alpha}.$ For each $i\in F$, we have
\begin{eqnarray*}
&&\sum_{j\in D} d_j x^*_{ij}=\sum_{j\in D} (d_j \sum_{i'\in N} x^1_{ii'} x^0_{i'j})=\sum_{j\in D} \sum_{i'\in N} (d_j x^1_{ii'} x^0_{i'j})\\
&&=\sum_{i'\in N} \sum_{j\in D} (d_j x^1_{ii'} x^0_{i'j})=\sum_{i'\in N} x^1_{ii'} \sum_{j\in D} (d_j x^0_{i'j})=\sum_{i'\in N} x^1_{ii'} d^1_{i'}\\
&&\leq (2+\frac{2}{\alpha})M y^1_{i}=(2+\frac{2}{\alpha})M y^*_{i},
\end{eqnarray*}
where the fifth equality holds by the definition of demand $d^1;$ the inequality holds by Theorem \ref{lemma:main theorem for CKM}.
\end{proof}

\subsection{Analysis}
\begin{lemma}
For any $\alpha\geq 4,$ there is a $(13+20\alpha)$-approximation algorithm for the hard uniform capacitated $k$-median location problem (CKL) by increasing the capacity by factor $2+\frac{2}{\alpha}.$
\end{lemma}
\begin{proof}
Let $COST(\cdot,\cdot)$ be the total cost of solution $(\cdot,\cdot)$. Let $OPT_{CKL}$ and $OPT_{CKM}$ be the optimal objective value of our original instance and constructed instance respectively.

By the process to obtain the constructed instance, we have $OPT_{CKM}\leq OPT_{CKL}+COST(x^0,y^0).$ Then,
\begin{eqnarray*}
&&COST(x^*,y^*)\\
&&\leq COST(x^1,y^1)+COST(x^0,y^0)\\
&&\leq (6+10\alpha)OPT_{CKM}+COST(x^0,y^0)\\
&&\leq (6+10\alpha)(OPT_{CKL}+COST(x^0,y^0))+COST(x^0,y^0)\\
&&\leq (13+20\alpha)OPT_{CKL},
\end{eqnarray*}
where the first inequality holds according to the process to get the solution $(x^*,y^*)$ and triangle inequalities; the second inequality follows by Theorem \ref{lemma:main theorem for CKM}; the last inequality holds as $COST(x^0,y^0)\leq OPT_{CKL}$.
\end{proof}

\section{Proof of Lemma \ref{lemma:property 1c}}\label{appendix:property 1c}
\begin{proof}
First, we show that location $i$ belongs to cluster $M_l$ if $c_{il}\leq \alpha C_l.$
For contradiction, suppose for some $i\in N$ with $c_{il}\leq \alpha C_l,$ $i\in M_{l'}$ instead of $i\in M_l$, where $l'\in N'-\{l\}$. This means $c_{il'}\leq c_{il}$ as we add $i$ to cluster $M_{l'}$ only if $N'(i)=l'$. Then, we have
\[
 c_{ll'}\leq c_{il}+c_{il'}\leq 2c_{il} \leq 2\alpha C_l,
\]
which is a contradiction as $c_{ll'}>2\alpha C_l$ by property \textbf{1b}.

Then, note that the total opening value of the locations, which are strictly greater than $\alpha C_l$ away from $l$ and serve some demand of $l$, should be strictly less than $\frac{1}{\alpha}$. That is, $\sum_{i\in N : c_{il}>\alpha C_l} x_{il}<\frac{1}{\alpha}$, otherwise $\sum_{i\in N}{c_{il} x_{il}}>C_l$ which is a contradiction.

So, $Z_l\geq \sum_{i\in N : c_{il}\leq \alpha C_l} y_{i}\geq \sum_{i\in N : c_{il}\leq \alpha C_l} x_{il} \geq \frac{\alpha-1}{\alpha},$ as $\sum_{i\in N} x_{il}=1$ and $x_{il}\leq y_i$ for each $i\in N.$
\end{proof}

\section{Proof of Lemma \ref{lemma:Move}}\label{appendix:Move}
\begin{proof}
It is easy to see that the Procedure 2 preserves (1) and (2).

To prove (3), it is sufficient to show that $\sum_{j\in N} d_j x'_{ij}\leq M y'_i$ holds for $i=j_a,j_b$
as we only change the demands served by $j_a$ and $j_b$.

For $j_b$, before performing the procedure we have
\[
 \sum_{j\in N} d_j x'_{j_b j}\leq M y'_{j_b}.
\]
Then, we get
\[
  (1-\frac{\delta}{y'_{j_b}})\sum_{j\in N} d_j x'_{j_b j}\leq M y'_{j_b}(1-\frac{\delta}{y'_{j_b}})=M (y'_{j_b}-\delta),
\] as $1-\frac{\delta}{y'_{j_b}}\geq 0$. Note that in the procedure we set
$x'_{j_b j}:=x'_{j_bj}-\frac{\delta}{y'_{j_b}} x'_{j_b j}$ and $y'_{j_b}:=y'_{j_b}-\delta$. So the lemma holds for $j_b$.

For $j_a$, before performing the procedure we have
\[
 \sum_{j\in N} d_j x'_{j_a j}\leq M y'_{j_a}.
\]
Then, we get
\[
  \sum_{j\in N} d_j x'_{j_a j}+\frac{\delta}{y'_{j_b}}\sum_{j\in N} d_j x'_{j_b j}\leq M y'_{j_a}+\frac{\delta}{y'_{j_b}} M y'_{j_b}=M(y'_{j_a}+\delta).
\] Recall that we set $x'_{j_a j}:=x'_{j_a j}+\frac{\delta}{y'_{j_b}} x'_{j_b j}$ and $y'_{j_a}:=y'_{j_a}+\delta$ in the procedure. Thus, the lemma also holds for $j_a$.
\end{proof}

\section{Proof of Lemma \ref{lemma:Concentrate}[2c]}\label{appendix:Concentrate}
\begin{proof}
Property \textbf{2c.} Initially, we set $y'_i=y_i, x'_{ij}=x_{ij}$ for all $i,j\in N$. Thus, $x'_{ij} \leq y'_i$ holds, for each $i, j\in N.$ We will show that after the procedure these inequalities still hold.

If $M_l$ is a non-terminal cluster, we have $0<\sum_{j\in M_l}y_j<1$.
Thus, $y'_{j_a}+y'_{j_b}<1$ when we perform Procedure 2, i.e., $y'_{j_b}<1-y'_{j_a}$.
Then, $\delta=y'_{j_b}$. So, we always set $y'_{j_b}:=0$, and $x'_{j_b j}:=0$ for each $j\in N$.
Thus, $x'_{j_b j}\leq y'_{j_b}$ for each $j\in N$.

For $j_a$, initially we have $x'_{j_a j}\leq y'_{j_a}$ for each $j\in N$.
Then, we get
\[x'_{j_a j}+\frac{\delta}{y'_{j_b}}x'_{j_b j}\leq y'_{j_a}+\frac{\delta}{y'_{j_b}}x'_{j_b j}\leq y'_{j_a}+\delta,\]
where the last inequality holds as initially we also have $x'_{j_b j}\leq y'_{j_b}.$
Thus, after Procedure 2 we still have $x'_{j_a j}\leq y'_{j_a}, \forall j\in N$ because we set $x'_{j_a j}:=x'_{j_a j}+\frac{\delta}{y'_{j_b}} x'_{j_b j}$ and $y'_{j_a}:=y'_{j_a}+\delta$ in the procedure.

Otherwise, $M_l$ is a terminal cluster. For this case, after Procedure 3 $y'_{i}=0$ or $y'_{i}\geq 1$ for each $i\in M_l.$ Note that $\sum_{i\in N} x'_{ij}=1, \forall j\in N$ always hold in the procedure (Lemma \ref{lemma:Move}). Thus, we have $x'_{ij}\leq y'_{i},\forall j\in N$ for each location $i$ with $y'_{i}\geq 1$.

Observe that in Procedure 2, only if $\delta=y'_{j_b}$ then we set $y'_{j_b}:=0.$ Meanwhile, if $\delta=y'_{j_b}$, we always set $x'_{j_b j}:=0$ for each $j\in N$. In the step 4 of Procedure 3, if we set $y'_{j_a}:=0,$ then set $x'_{j_a j}:=0$ for each $j\in N$.
Thus, $x'_{ij}\leq y'_{i},\forall j\in N$ for each location $i$ with $y'_{i}=0$.
\end{proof}

\section{Proof of Lemma \ref{lemma:extra cost of step 2}}\label{appendix:extra cost of step 2}
\begin{proof}
The idea is similar as that in \cite{Charikar,Guha}. (1) In the procedure, if $M_l$ is a non-terminal cluster, we close all locations in $M_l-\{l\}$, and assign the demands originally served by location $j_b\in M_l$ to the cluster core $l$.

Moreover, we have $c_{l j_b}\leq c_{N'(j) j_b}$ as $N'(j_b)=l$ is the closest location to $j_b$ in $N'$ and $N'(j)\in N'.$ Then,
\begin{eqnarray*}
c_{l j}&&\leq c_{j_b j}+c_{l j_b}\\
&& \leq c_{j_b j}+c_{N'(j) j_b}\\
&& \leq c_{j_b j}+(c_{j_b j }+c_{N'(j) j})\\
&& \leq 2c_{j_b j}+2\alpha C_j,
\end{eqnarray*}
where the first and third inequality hold by triangle inequalities; the last inequality follows by property \textbf{1a}.

(2) If we move the demand of location $j$ served by $j_b$ to $j_a$, then we know $c_{l j_a}\leq c_{l j_b}$ by the procedure. Note that we also have $c_{l j_b}\leq c_{N'(j) j_b}$ for this case. Then,
\begin{eqnarray*}
c_{j_a j}&&\leq c_{j_a l}+c_{l j_b} +c_{j_b j}\\
&& \leq 2c_{l j_b} +c_{j_b j}\\
&& \leq 2c_{N'(j) j_b} +c_{j_b j}\\
&& \leq 2(c_{N'(j) j}+c_{j j_b}) +c_{j_b j}\\
&& \leq 4\alpha C_j+3c_{j_b j},
\end{eqnarray*}
where the first and fourth inequality hold by triangle inequalities; the last inequality follows by property \textbf{1a}.
\end{proof}

\section{The Details of Remark \ref{remark:step 3 terminates}}\label{appendix:step 3 terminates}
The Procedure 4 terminates at $r>1$. Note that we only consider the case $Y\leq |N_2|-1$, i.e., $\sum_{t=1}^v y'_{i_t}\leq |N_2|-1.$ For contradiction, suppose the procedure terminates at $r=1$, which means $\hat{y}_{i_t}=1$ for $t=2,\cdots,v.$ So, $\sum_{t=2}^v \hat{y}_{i_t}=v-1=|N_2|-1.$ Moreover, we have $\hat{y}_{i_1}>0$, since initially we set $\hat{y}_{i_1}:=\frac{\alpha-1}{\alpha}, \alpha\geq 4,$ and later we move strictly less than $\frac{1}{\alpha}$ opening value from $i_1$ to other location $t$ to make $\hat{y}_{i_t}:=1$ if necessary. Thus, we have $\sum_{t=1}^v \hat{y}_{i_t}=v-1+\hat{y}_{i_1}>|N_2|-1.$ Then since $\sum_{i\in N_2}{\hat{y}_i}=\sum_{i\in N_2}{y'_i}$ (Lemma \ref{lemma:Determine new opening values}), we get $\sum_{t=1}^v y'_{i_t}> |N_2|-1,$ a contradiction.

\section{Proof of Lemma \ref{lemma:Determine new opening values}}\label{appendix:Determine new opening values}
\begin{proof}
Property \textbf{3a}. For each location $i\in N-N_2$, we set $\hat{y}_i:=y'_i$. So, $1\leq \hat{y}_i< 2$ for each $i\in N_1$; $\hat{y}_i=0$ for each $i\in N-N_1\cup N_2$.

 For each location $i\in N_2$, initially we set $\hat{y}_{i}:=\frac{\alpha-1}{\alpha}$. In the procedure, we have two cases for location $i_r$ in $N_2$. First, if $Y'>0$ and $Y'+\hat{y}_{i_r}\geq 1, i_r\in N_2$, we always set $\hat{y}_{i_r}:=1$. So it is sufficient to check the case $Y'>0$ and $Y'+\hat{y}_{i_r}< 1$.
 For this case, we set $\hat{y}_{i_1}:=\hat{y}_{i_1}-(1-Y'-\hat{y}_{i_r}), \hat{y}_{i_r}:=1$ and terminates. Thus, we have $\hat{y}_{i_t}:=1$ for $t\geq r$, and $\hat{y}_{i_t}:=\frac{\alpha-1}{\alpha}$ for $2\leq t< r$. Moreover, we get $\frac{\alpha-2}{\alpha}< \hat{y}_{i_1}< \frac{\alpha-1}{\alpha}$ as $0<1-Y'-\hat{y}_{i_r}<\frac{1}{\alpha}$ (note that $\hat{y}_{i_r}=\frac{\alpha-1}{\alpha}$ before we set it to be $1$).

 Thus, for all $i\in N$, $\frac{\alpha-2}{\alpha}< \hat{y}_i\leq \frac{\alpha-1}{\alpha}$, or $1\leq \hat{y}_i< 2,$ or $\hat{y}_i=0;$ and
 $|\{i\in N \mid \frac{\alpha-2}{\alpha}< \hat{y}_i < \frac{\alpha-1}{\alpha}\}|\leq 1.$

 Property \textbf{3b, 3C.} By property \textbf{2a}, for all $i\in N$, $\frac{\alpha-1}{\alpha}\leq y'_i<2$ or $y'_i=0$; and $\sum_{j\in N} d_j x'_{ij}\leq M y'_i$.

 Notice that if for location $i$ we have $\frac{\alpha-2}{\alpha}< \hat{y}_i\leq \frac{\alpha-1}{\alpha}$ after the procedure, then we know $\frac{\alpha-1}{\alpha}\leq y'_i<1.$ Otherwise, $\hat{y}_i=0$ or $\hat{y}_i\geq 1$, a contradiction. Moreover, we make no change on $x'$. So we have if $\frac{\alpha-2}{\alpha}< \hat{y}_i\leq \frac{\alpha-1}{\alpha}$, then $\sum_{j\in N} d_j x'_{ij}\leq M y'_i< M.$

 If $1\leq  \hat{y}_i<2$ for location $i$, then we have $y'_i\leq \hat{y}_i$ as in the procedure only for the following two cases we will set $1\leq  \hat{y}_i<2$. For case 1: $1\leq y'_i<2$, we set $\hat{y}_i:= y'_i$. For case 2: $0< y'_i<1, i_r=i$ and $Y'> 0$, we set $\hat{y}_i=1$. For both cases, we have $y'_i\leq \hat{y}_i$.
 Thus, if $1\leq \hat{y}_i< 2$, then $\sum_{j\in N} d_j x'_{ij}\leq M y'_i\leq  M \hat{y}_i.$

Property \textbf{3d}. Recall that $1> y'_i\geq \frac{\alpha-1}{\alpha}, \forall i\in N_2$. So, we have:
$$\sum_{i\in N_2} y'_i d'_i c_{s(i) i}=\sum_{i\in N_2} (y'_i-\frac{\alpha-1}{\alpha}) d'_i c_{s(i) i}+\sum_{i\in N_2} \frac{\alpha-1}{\alpha} d'_i c_{s(i) i}.$$

In the process, initially we set $\hat{y}_i=\frac{\alpha-1}{\alpha},$ for all $i\in N_2.$
Then, we use $Y'=\sum_{i\in N_2} (y'_i-\hat{y}_i)$ to make the opening value of the location with biggest value $d'_i c_{s(i) i}, i\in N_2$, i.e., $i_v$, to be $1$.
If $\hat{y}_{i_t}=1$ for all $v\geq t\geq r$, then we try to make $\hat{y}_{i_{r-1}}$ to be 1 until $Y'+\hat{y}_{i_{r}}<1$.
If $Y'+\hat{y}_{i_{r}}<1$, then we reduce $\hat{y}_{i_1}$ by $1-Y'-\hat{y}_{i_r}$ to make $\hat{y}_{i_r}$ to be 1. That is, we only move the opening value from one location to the other locations.
So, $\sum_{i\in N_2}{\hat{y}_i}=\sum_{i\in N_2}{y'_i}$ holds after this procedure.  Moreover, we set $\hat{y}_i:=y'_i$ for each $i\in N-N_2$. Thus, we also have $\sum_{i\in N}{\hat{y}_i}=\sum_{i\in N}{y'_i}\leq k.$

Property \textbf{3e}. Initially, the value of $\sum_{i\in N_2} \hat{y}_i d'_i c_{s(i) i}$ is $\sum_{i\in N_2} \frac{\alpha-1}{\alpha} d'_i c_{s(i) i}.$ Later, we transfer value $\delta$ from $\{i_1,\cdots,i_{r-1}\}$ to $i_r$ to make  $\hat{y}_{i_r}:=\hat{y}_{i_r}+\delta$ until $\hat{y}_{i_r}=1$.
For each transfer operation, we increase the value $\sum_{i\in N_2} \hat{y}_i d'_i c_{s(i) i}$ by
$\delta d'_{i_r} c_{s({i_r}) {i_r}}$, which is greater than or equal to $\delta d'_{i} c_{s({i}) {i}}$ for each $i\in\{i_1,\cdots,i_{r-1}\}.$

Moreover, by property \textbf{3d}, we know $\sum_{i\in N_2}{\hat{y}_i}=\sum_{i\in N_2}{y'_i}$ finally.
Therefore, $\sum_{i\in N_2} \hat{y}_i d'_i c_{s(i) i} \geq \sum_{i\in N_2} y'_i d'_i c_{s(i) i}$ after this procedure. Then, we have $\sum_{i\in N_2}(1-\hat{y}_i)d'_i c_{s(i) i} \leq \sum_{i\in N_2}(1-y'_i)d'_i c_{s(i) i}.$
\end{proof}

\section{Proof of Lemma \ref{lemma:total fractional opening value}}\label{appendix:total fractional opening value}
\begin{proof}
(1)By property \textbf{3a}, $|\{i\in N \mid \frac{\alpha-2}{\alpha}< \hat{y}_i < \frac{\alpha-1}{\alpha}\}|\leq 1.$ So in $\hat{N}_2$ at most one location has a fractional opening value in $(\frac{\alpha-2}{\alpha},\frac{\alpha-1}{\alpha})$, and all other locations have fractional opening value exactly equal to $\frac{\alpha-1}{\alpha}$.

So,
\[\sum_{i\in Q_t\cap \hat{N}_2} \hat{y}_i> \frac{\alpha-2}{\alpha}+\frac{\alpha-1}{\alpha} (2q-1)= \frac{2q\alpha-2q-1}{\alpha}= q+\frac{q\alpha-2q-1}{\alpha}.
\]

Moreover, since $\alpha\geq 4$ and $q\geq 1$, we have $\frac{q\alpha-2q-1}{\alpha}\geq \frac{2q-1}{\alpha}> 0.$ Thus, $\sum_{i\in Q_t\cap \hat{N}_2} \hat{y}_i> q.$

(2) First, we have \[\sum_{i\in Q_t\cap \hat{N}_2} \hat{y}_i> \frac{\alpha-2}{\alpha}+\frac{\alpha-1}{\alpha} 2q= \frac{2q\alpha-2q+\alpha-2}{\alpha}= q+1+\frac{q\alpha-2q-2}{\alpha}.
\]

Then, as $\alpha\geq 4$ and $q\geq 1$, we get $\frac{q\alpha-2q-2}{\alpha}\geq \frac{2q-2}{\alpha}\geq 0.$ Thus, $\sum_{i\in Q_t\cap \hat{N}_2} \hat{y}_i> q+1.$
\end{proof}

\section{Proof of Lemma \ref{lemma:the number of centers of each star}}\label{appendix:the number of centers of each star}
\begin{proof}
1. For the case: an even star rooted at location $t$ with $1\leq \hat{y}_t<2$, we totally build $q+1$ centers at location: $t, i_1,i_3,\cdots,i_{2q-1}$. By Lemma \ref{lemma:total fractional opening value}, we have
\[q+1\leq \lfloor \sum_{i\in Q_t\cap \hat{N}_2} \hat{y}_i \rfloor+1\leq \lfloor R_t \rfloor,\]
where the last inequality holds as $R_t=\sum_{i\in Q_t\cap \hat{N}_2} \hat{y}_i+\hat{y}_t$ and $\hat{y}_t\geq 1.$

We serve the demand $d'_{i_{2r}}$ of $i_{2r}$ by opened location $i_{2r-1}$. Then, the total demand served by $i_{2r-1}$ is $d'_{i_{2r}}+d'_{i_{2r-1}}.$ By property \textbf{3b}, we get $\sum_{j\in N}{d'_j \bar{x}_{i_{2r-1}j}}= d'_{i_{2r}}+d'_{i_{2r-1}}\leq 2M< (2+\frac{2}{\alpha})M \bar{y}_{i_{2r-1}}$. By property \textbf{3c},
for the root $t$ we have $\sum_{j\in N}{d'_j \bar{x}_{tj}}=d'_{t}< 2M< (2+\frac{2}{\alpha})M \bar{y}_{t}$.

2. For the case: an even star $Q_t$ rooted at location $t$ with $\frac{\alpha-2}{\alpha}< \hat{y}_t\leq \frac{\alpha-1}{\alpha}$, we build $q+1$ centers. By Lemma \ref{lemma:total fractional opening value}, we have $q+1 \leq \lfloor R_t \rfloor$ as $|Q_t\cap \hat{N}_2|=2q+1$. Again by property \textbf{3b}, we can obtain for each $i\in Q_t$, $\sum_{j\in N}{d'_j \bar{x}_{ij}}\leq (2+\frac{2}{\alpha})M \bar{y}_{i}$.

3. For the case: an odd star rooted at location $t$ with $1+\frac{2}{\alpha}\leq \hat{y}_t<2$,
we build a center at location: $t,i_1,i_3,\cdots,i_{2q+1}$, and serve the demand $d'_{i_{2r}}$ of $i_{2r}$ by opened location $i_{2r-1}$.

If $q\geq 1$, we build $q+2$ centers. By Lemma \ref{lemma:total fractional opening value}, we have
\[q+2\leq \lfloor \sum_{i\in Q_t\cap \hat{N}_2} \hat{y}_i \rfloor +1\leq \lfloor R_t \rfloor,\]
where the first inequality holds as $|Q_t\cap \hat{N}_2|=2q+1$, the last inequality holds as $R_t=\sum_{i\in Q_t\cap \hat{N}_2} \hat{y}_i+\hat{y}_t$ and $\hat{y}_t\geq 1.$
By property \textbf{3b} and \textbf{3c}, for each $i\in Q_t$, we have $\sum_{j\in N}{d'_j \bar{x}_{ij}}\leq (2+\frac{2}{\alpha})M \bar{y}_{i}$.

If $q=0$, i.e., the root $t$ has only one child $i_1$, we build two centers at location $t$ and $i_1$.
Thus, we build at most $\lfloor R_t \rfloor$ centers as $R_t=\hat{y}_t+\hat{y}_{i_1}> 1+\frac{2}{\alpha}+\frac{\alpha-2}{\alpha}=2$. Moreover, we serve the demand of $t$ and $i_1$ by themselves. So the relaxed capacity constraint $\sum_{j\in N}{d'_j \bar{x}_{ij}}\leq (2+\frac{2}{\alpha})M \bar{y}_i$ holds by property \textbf{3b} and \textbf{3c}.

4. For the case: an odd star rooted at location $t$ with $\frac{\alpha-2}{\alpha}< \hat{y}_t\leq \frac{\alpha-1}{\alpha}$ or $1\leq \hat{y}_t< 1+\frac{2}{\alpha}$,
we build a center at location: $t,i_2,i_4,\cdots,i_{2q}$, and serve the demand $d'_{i_{2r+1}}$ of $i_{2r+1}$ by opened location $i_{2r}$, serve the demand $d'_{i_{1}}$ of $i_{1}$ by $t$.

First, we consider $1\leq \hat{y}_t< 1+\frac{2}{\alpha}$.
If $q\geq 1$, we build $q+1$ centers. By Lemma \ref{lemma:total fractional opening value}, we have
\[q+1\leq \lfloor \sum_{i\in Q_t\cap \hat{N}_2} \hat{y}_i \rfloor \leq \lfloor R_t \rfloor,\]
where the last inequality holds as $R_t=\sum_{i\in Q_t\cap \hat{N}_2} \hat{y}_i+\hat{y}_t$ and $\hat{y}_t\geq 1.$
By property \textbf{3b}, for each $i_{2r}$ we have $\sum_{j\in N}{d'_j \bar{x}_{i_{2r}j}}= d'_{i_{2r}}+d'_{i_{2r+1}}\leq 2M< (2+\frac{2}{\alpha})M \bar{y}_{i_{2r}}$.
By property \textbf{3b} and \textbf{3c},
for the root $t$ we have $\sum_{j\in N}{d'_j \bar{x}_{tj}}=d'_{t}+d'_{i_{1}}\leq (1+\frac{2}{\alpha})M+M\leq (2+\frac{2}{\alpha})M \bar{y}_{t}$.

If $q=0$, i.e., the root $t$ has only one child $i_1$, we build a center at location $t$ and serve the demand of $i_1$ by $t$. Thus, we build at most $\lfloor R_t \rfloor$ centers as $1\leq R_t$. Moreover, $\sum_{j\in N}{d'_j \bar{x}_{tj}}=d'_t+d'_{i_1}\leq (2+\frac{2}{\alpha})M \bar{y}_{t}$, by property \textbf{3b} and \textbf{3c}.

Then, we consider $\frac{\alpha-2}{\alpha}< \hat{y}_t\leq \frac{\alpha-1}{\alpha}$.
If $q\geq 1$, we open $q+1$ centers. By Lemma \ref{lemma:total fractional opening value}, we have
\[q+1\leq \lfloor \sum_{i\in Q_t\cap \hat{N}_2} \hat{y}_i \rfloor,\]
as $|Q_t\cap \hat{N}_2|=2q+2.$
Again by property \textbf{3b} and \textbf{3c}, for each $i\in Q_t$ we have $\sum_{j\in N}{d'_j \bar{x}_{ij}}\leq (2+\frac{2}{\alpha})M \bar{y}_i$.

If $q=0$, i.e., the root $t$ has only one child $i_1$, we build a center at location $t$ and serve the demand of $i_1$ by $t$. Thus, we build at most $\lfloor R_t \rfloor$ centers as $R_t=\hat{y}_t+\hat{y}_{i_1}>\frac{\alpha-2}{\alpha}+\frac{\alpha-1}{\alpha}=1+\frac{\alpha-3}{\alpha}>1 (\alpha\geq 4)$, where the first inequality holds by property \textbf{3a}. Moreover, $\sum_{j\in N}{d'_j \bar{x}_{tj}}=d'_t+d'_{i_1}\leq (2+\frac{2}{\alpha})M \bar{y}_{t}$, by property \textbf{3b}.
\end{proof}

\section{Proof of Lemma \ref{lemma:extra cost of step 4}}\label{appendix:extra cost of step 4}
\begin{proof}
Note that in this proof we only consider location $i\in \hat{N}_2$, since we always build a center at each location in $\hat{N}_1$ and serve its demand by itself.

For each star $Q_t$, the reassignment is always to serve the demand $d'_{i}$ of location $i$ by an opened location $i'$ that is closer to the root $t$, where $i,i'\in Q_t$ and $c_{ti'}\leq c_{ti}.$ Recall that $s(i)$ is the closest location to $i$ in $({N}_1\cup {N}_2)-\{i\}$. By Procedure 5, we know $s(i)=s(i')=t.$
The cost for this reassignment is $d'_{i}c_{i' i}$, which can be bounded by $2d'_{i}c_{s(i)i}$ as $c_{i' i}\leq c_{s(i) i'}+c_{s(i)i}\leq 2c_{s(i)i}.$

Since $\frac{\alpha-2}{\alpha}< \hat{y}_i\leq \frac{\alpha-1}{\alpha}$ for each $i\in Q_t\cap \hat{N}_2$, we have
\[2d'_{i}c_{s(i)i}\leq 2\alpha (1-\hat{y}_i)d'_{i}c_{s(i)i}.\] We sum $2\alpha (1-\hat{y}_i)d'_{i}c_{s(i)i}$ over all $i\in \hat{N}_2$ to get an upper bound for the total cost of stars, i.e., $\sum_{i\in \hat{N}_2}2\alpha (1-\hat{y}_i)d'_{i}c_{s(i)i}.$
Note that $\hat{N}_2\subseteq N_2,$ and $\hat{y}_i\leq 1$ for each $i\in N_2$. Thus,
 \[\sum_{i\in \hat{N}_2}2\alpha (1-\hat{y}_i)d'_{i}c_{s(i)i}\leq \sum_{i\in {N}_2}2\alpha (1-\hat{y}_i)d'_{i}c_{s(i)i}.\]
Further by property \textbf{3e}, we know
\[\sum_{i\in {N}_2}2\alpha (1-\hat{y}_i)d'_{i}c_{s(i)i}\leq \sum_{i\in {N}_2}2\alpha (1-y'_i)d'_{i}c_{s(i)i}.\]
By the definition of $d'_i$, we have
\[\sum_{i\in {N}_2} 2\alpha (1-y'_i)d'_{i}c_{s(i)i}=\sum_{i\in {N}_2} \sum_{j\in N}2\alpha (1-y'_i)d_{j}x'_{ij}c_{s(i)i}.\]
Recall that for each $i\in N_2$, the core of non-terminal cluster $M_i$, we move the amount of each location in $M_i$ to the core $i$ by Procedure 3 (Lemma \ref{lemma:extra cost of step 2}). That is, $x'_{ij}=\sum_{i'\in M_i} x_{i'j}.$ So,
\[\sum_{i\in {N}_2} \sum_{j\in N}2\alpha (1-y'_i)d_{j}x'_{ij}c_{s(i)i}=\sum_{i\in {N}_2} \sum_{j\in N} \sum_{i'\in M_i} 2\alpha (1-y'_i)d_{j}x_{i'j}c_{s(i)i}.\]

Therefore, it is sufficient to show that
\[\sum_{i\in {N}_2} \sum_{j\in N} \sum_{i'\in M_i} 2\alpha (1-y'_i)d_{j}x_{i'j}c_{s(i)i}\leq \sum_{i\in {N}_2}\sum_{j\in N} \sum_{i'\in M_i} d_j(4c_{i'j}x_{i'j}+8\alpha C_j x_{i'j} ).\]
That is, it is sufficient to show that for each $j\in N, i'\in M_i, i\in N_2$
\[2\alpha (1-y'_i)d_{j}x_{i'j}c_{s(i)i}\leq d_j(4c_{i'j}x_{i'j}+8\alpha C_j x_{i'j} ).\]

We have two cases: (a) $N'(j)=i$ and (b) $N'(j)\neq i.$ We show the above inequality holds for both cases.

\noindent \textbf{(a) $N'(j)=i.$}

Since $y'_i\in [\frac{\alpha-1}{\alpha},1), \forall i\in N_2$ (i.e., $\frac{\alpha-1}{\alpha}\leq y'_i=Z_i=\sum_{i'\in M_i} y_{i'}<1$), there exists a location $i^*\notin M_i$ such that $x_{i^*i}>0$ (Recall that we have $\sum_{r\in N} x_{ri}=1$ and $x_{ri}\leq y_r$).
More precisely, we can find a location $i^*\notin M_i$ with $x_{i^*i}>0$ and $c_{i^*i}\leq \frac{C_i}{1-y'_i}.$
Otherwise, $\sum_{r\in N} x_{ri}c_{ri}> C_i,$ a contradiction.

Since $N'(i^*)\neq i$, $c_{N'(i^*)i^*}\leq c_{ii^*}$ as $N'(i^*)$ is the closest location to $i^*$ in $N'$ and $i\in N'.$ So,
\[c_{s(i)i}\leq c_{iN'(i^*)}\leq c_{N'(i^*)i^*}+ c_{ii^*}\leq 2c_{ii^*}\leq 2\frac{C_i}{1-y'_i},\]
where the first inequality holds as $s(i)$ is the closest location to $i$ in $N_1\cup N_2-\{i\}$
and $N'(i^*)\in (N_1\cup N_2)-\{i\};$ the second inequality holds by triangle inequality.

If $C_i\leq C_j$, then we have
\begin{equation}\label{bound for case a1}
2\alpha (1-y'_i)d_{j}x_{i'j}c_{s(i)i}\leq  2\alpha d_{j}x_{i'j}2 C_i \leq 4\alpha d_{j}x_{i'j} C_j.
\end{equation}

Otherwise $C_i> C_j$. We will show that if $C_i> C_j$, then $C_i\leq 2C_j$.
Since $C_i>C_j$, we consider location $j$ before $i$ when we choose the cluster cores $N'$.
Then, $j$ can not be a cluster core. Otherwise, $i$ cannot be a cluster core as $c_{ij}\leq 2\alpha C_j<2\alpha C_i$ (by property \textbf{1a}), a contradiction.
This means there already exists a location $r\in N'$ with $C_r\leq C_j$ and $C_{rj}\leq 2\alpha C_j$ before we check whether $j$ should be chosen as a cluster core.
Note that $c_{ri}>2\alpha C_i$ as $i$ is also chosen as a cluster core and property \textbf{1b}.
Moreover, by triangle inequality $c_{rj}+c_{ij}\geq c_{ri}$.
So, $2\alpha C_i<c_{ri}\leq c_{rj}+c_{ij}\leq 2\alpha C_j+2\alpha C_j=4\alpha C_j.$ That is, $C_i\leq 2 C_j.$
Thus, for this case we have
\begin{equation}\label{bound for case a2}
2\alpha (1-y'_i)d_{j}x_{i'j}c_{s(i)i}\leq  2\alpha d_{j}x_{i'j}2 C_i \leq 8\alpha d_{j}x_{i'j} C_j.
\end{equation}

\noindent \textbf{(b) $N'(j)\neq i.$}

The proof for this case is similar as that in \cite{Charikar,Guha}. For each $i'\in M_i$, $N'(i')=i.$ If $N'(j)\neq i$, then $c_{i'i}=c_{i'N'(i')}\leq c_{i'N'(j)}$ as $N'(i')$ is the closest location to $i'$ in $N'$ and $N'(j)\in N'.$
Thus, we have
\[c_{s(i)i}\leq c_{N'(j)i}\leq c_{i'i}+c_{i'N'(j)}\leq 2 c_{i'N'(j)}\leq 2 (c_{i'j}+c_{N'(j)j}),\]
where $i'\in M_i;$ the first inequality holds as $s(i)$ is the closest location to $i$ in $(N_1\cup N_2)-\{i\}$
and $N'(j)\in (N_1\cup N_2)-\{i\}$ (note that $N_1\cup N_2 \supseteq N'$); the second and fourth inequalities hold by triangle inequalities.

By property \textbf{1a}, $c_{N'(j)j}\leq 2\alpha C_j.$ So,
\[c_{s(i)i}\leq 2c_{i'j}+4\alpha C_j.\]

Note that $0<\alpha (1-y'_i)\leq 1$ as $1> y'_i\geq \frac{\alpha-1}{\alpha}, i\in N_2.$ So, we have
\begin{equation}\label{bound for case b}
2\alpha (1-y'_i)d_{j}x_{i'j}c_{s(i)i}\leq  2d_{j}x_{i'j}(2c_{i'j}+4\alpha C_j)=d_{j}(4c_{i'j}x_{i'j}+8\alpha C_j x_{i'j}).
\end{equation}

From inequalities (\ref{bound for case a1}), (\ref{bound for case a2}) and (\ref{bound for case b}), we get
$$2\alpha (1-y'_i)d_{j}x_{i'j}c_{s(i)i}\leq d_{j}(4c_{i'j}x_{i'j}+8\alpha C_j x_{i'j}).$$
\end{proof}
\end{appendix}

\end{document}